\documentclass{amsart}
\usepackage{graphicx} 
 \usepackage{subcaption}
\usepackage{xcolor}
\usepackage{tikz, pgfplots}
\usetikzlibrary{decorations.fractals}
\usetikzlibrary{intersections}
\usepgfplotslibrary{fillbetween} 
\usetikzlibrary{lindenmayersystems}
\usepackage{amsmath, amssymb}
\newtheorem{lemma}{Lemma}
\newtheorem{theorem}{Theorem}
\newtheorem{observation}{Observation}

\newcommand{\V}{\mathcal{V}}
\newcommand{\Sw}{\mathcal{S}}

\title{A greedy maximal sweepline algorithm for a Jordan curve}
\author[Apurva Mudgal]{Apurva Mudgal, \\ Department of Computer Science and Engineering, \\ Indian Institute of Technology Ropar, \\ Rupnagar, Punjab - 140001, \\ Email: {\tt apurva@iitrpr.ac.in}}

\begin{document}

\maketitle
\begin{abstract}
We give a greedy sweepline algorithm for a Jordan curve and prove that it is maximal in the sense of \cite{apurva}. Our proof uses K\H{o}nig's lemma.
\end{abstract}

\tableofcontents

Jordan curve theorem is a fundamental theorem in topology, with several known proofs \cite{hales, compendium1, compendium2}. In \cite{apurva}, the author gave a proof of Jordan curve theorem based on the sweepline algorithm for the trapezoidal decomposition of a polygon (see the book \cite{overmars}). The proof uses Zorn's lemma to prove the existence of a maximal sweepline algorithm. It does not specify a maximal sweepline algorithm for a Jordan curve.

In this short note, we remedy this drawback. We describe a greedy sweepline algorithm for a Jordan curve, and prove that it is maximal. Instead of Zorn's lemma, the proof now uses K\H{o}nig's lemma \cite{diestel}.

 The greedy algorithm always extends the current region swept by picking a vertical segment $u_{max}$ of maximum length, among the currently available vertical segments. The horizontal segment for the next horizontal sweep starts at the midpoint of $u_{max}$ and is extended as far as possible in the exterior till it hits a point of the given Jordan curve $J$. This process is repeated again and again.

In the following, we assume familiarity with the paper \cite{apurva}.

\section{Maximum length vertical segments in finite trees of horizontal sweeps}

\begin{observation}
\label{obs-max-vertical-seg-horizontal-sweep}
Let $H(t)$ be a horizontal sweep and let $K_t$ be its boundary Jordan curve (see Lemma $2$ in \cite{apurva}). Let $\V(K_t)$ denote the set of vertical segments in $K_t$. Suppose $\V(K_t) \neq \phi$.
Then, there exists a vertical segment $u_{max} \in \V(K_t)$ such that the length of $u_{max}$ is greater than or equal to the length of any vertical segment in $\V(K_t)$.
\end{observation}

\noindent {\bf Proof:} Suppose $\V(K_t)$ is non-empty. Let $u'$ be {\it any} vertical segment in $\V(K_t)$. Let $\delta$ be the length of $u'$. Note that $\delta > 0$.

Let $W_{\delta}$ be the set of all vertical segments in $\V(K_t)$ with length greater than or equal to $\delta$. Since $u' \in W_{\delta}$, $W_{\delta} \neq \phi$.

We prove that $W_{\delta}$ is a finite set. Suppose, for the sake of contradiction, that $W_{\delta}$ is countably infinite. Then, there exists a convergent sequence $u_1, u_2, \ldots$ of vertical segments in $W_{\delta}$. Since each $u_i$, $i \in \mathbb{N}$, has length at least $\delta$, the limiting vertical segment $u^*$ of this sequence also has length at least $\delta$. Since $u_1, u_2, \ldots$ are mutually disjoint arcs of the Jordan curve $K_t$, we arrive at a contradiction (take $\gamma_i = u_i$, $i \in \mathbb{N}$, in Observation $4$ of \cite{apurva}). 

We conclude the proof by taking $u_{max}$ as the maximum length vertical segment in the finite set $W_{\delta}$.
$\blacksquare$

\begin{observation}
\label{obs-max-vertical-seg-finite-tree}
Let $\mathcal{S} = (H(t_1), H(t_2), \ldots, H(t_k))$ be a sweepline algorithm, consisting of a finite number of horizontal sweeps. Let $K$ be the piecewise-vertical Jordan curve forming its boundary (see Lemma $5$ in \cite{apurva}).
Let $\V(K)$ denote the set of vertical segments in $K$. Suppose $\V(K) \neq \phi$. Then, there exists a vertical segment $u_{max} \in \V(K)$ such that the length of $u_{max}$ is greater than or equal to the length of any
vertical segment in $\V(K)$.
\end{observation}

\noindent {\bf Proof:} On the same lines as the proof of Observation \ref{obs-max-vertical-seg-horizontal-sweep}. $\blacksquare$
\section{The greedy sweepline algorithm $\Sw_{greedy}$} 

\subsection{Algorithm $\Sw_{greedy}$.} We now describe the greedy sweepline algorithm $\Sw_{greedy}$:

\begin{enumerate}
\item As in \cite{apurva}, we assume that we are given a point $p^* \in \mathbb{R}^2 - J$ such that its equivalence class $[s_{p^*}]$ has
no open segments of infinite length. 

\item We choose the initial horizontal segment $t_1$ such that (i) $p^* \in int(t_1)$, (ii) $int(t_1) \cap J = \phi$ and (iii) both endpoints of $t_1$ belong to $J$. 
Let $E_1 = H(t_1)$ be the initial horizontal sweep and let $K_1$ be the piecewise-vertical Jordan curve forming its boundary.

\item For $i= 1, 2, \ldots$:

	\begin{enumerate}
		\item If $\V(K_i) = \phi$, halt.
		
		\item Let $u_i$ be a vertical segment in $\V(K_i)$ such that length of $u_i$ is greater than or equal to the length of any vertical
		segment in $\V(K_i)$. 
	
		\item Let $m_i$ be the {\bf midpoint} of vertical segment $u_i$.
		
		\item Let $t_{i+1}$ be a horizontal segment such that (i) one endpoint of $t_{i+1}$ is at $m_i$, (ii) the other endpoint of $t_{i+1}$ is on the 
	          Jordan curve $J$, (iii) $int(t_{i+1}) \cap J = \phi$, and (iv) $t_{i+1} - \{ m_i \} \subset ext(K_i)$. 
		
		\item $K_i$ is extended by $H(t_{i+1})$. Let $K_{i+1}$ be the piecewise-vertical Jordan curve bounding $E_{i+1} = (H(t_1), H(t_2), \ldots, H(t_{i+1}))$.
	\end{enumerate}
\end{enumerate}

\subsection{Algorithm $\Sw_{greedy}$ is well-defined}
{\it $t_1$ exists.} Let $t_1$ be the union of all horizontal segments $t'$ such that (i) $p^* \in int(t')$ and (ii) $int(t') \cap J = \phi$. Suppose, for the sake of contradiction, that $t_{1}$ is infinite. Let $q_{1}$ be a point of $int(t_{1})$ that lies outside the bounding box $\mathcal{B}$ of $J$. Then, $s_{q_1} \in [s_{p^*}]$
and $s_{q_1}$ is infinite (a contradiction).

{\it $u_i$ exists, for each $i \in \mathbb{N}$.} For each $i \in \mathbb{N}$, the vertical segment $u_i$ exists by Observation \ref{obs-max-vertical-seg-finite-tree}.

{\it $t_{i+1}$ exists, for each $i \in \mathbb{N}$.} For each $i \in \mathbb{N}$, let $t_{i+1}$ be the union of all horizontal segments $t'$ such that (i) one endpoint of $t'$ is at $m_i$, (ii) $int(t') \cap J = \phi$, and (iii) $t' - \{ m_i \} \subset ext(K_i)$. Suppose, for the sake of contradiction, that $t_{i+1}$ is infinite. Let $q_{i+1}$ be a point of $int(t_{i+1})$ that lies outside the bounding box $\mathcal{B}$ of $J$. Then, $s_{q_{i+1}} \in [s_{p^*}]$
and $s_{q_{i+1}}$ is infinite (a contradiction).

\section{Proof of maximality}

\begin{lemma}
\label{lem-no-inf-ntm}
The recursion tree $T(\Sw_{greedy})$ has no non-terminating {\bf infinite} ray $r$, such that $r$ starts at the root node of $T(\Sw_{greedy})$.
\end{lemma}

\noindent {\bf Proof:} Suppose, for the sake of contradiction, that there exists a non-terminating infinite ray $r=(H(t_{i_1}), H(t_{i_2}), \ldots)$ in $T(\Sw_{greedy})$, such
that $i_1=1$ and $i_1 < i_2 < \cdots$. 

Let $s^*_r$ be the unique limiting segment of $r$ (see Lemma $7$ of \cite{apurva}). Since $r$ is a non-terminating infinite ray, the length $|s^*_r|$ of $s^*_r$ is positive.
Let $m$ be the midpoint of $s^*_r$. 
Construct a horizontal line segment $t_{s^*_r}$ of positive length, as in the proof of Lemma $8$ of \cite{apurva}. One endpoint of $t_{s^*_r}$ is at $m$.
Then, there exists a natural number $N'$ and an open segment $s_{i_{N'}} \in H(t_{i_{N'}})$ in $E_{N'} =(H(t_1), H(t_2), \ldots, H(t_{N'}))$ such that the finite ray $r'$ obtained by replacing the portion of ray $r$ generated after open segment $s_{i_{N'}}$ by the horizontal sweep $H(t')$, where $t'$ is the portion of $t_{s^*_r}$ between $s_{i_{N'}}$ and $s^*_r$, is equivalent to $r$ i.e., $r \sim r'$ (see proof of Lemma $8$ of \cite{apurva}).

Since $s^*_r$ is a limiting segment of ray $r$, there exists a horizontal segment $t''$ such that (i) $t'' \subset t'$, (ii) one endpoint of $t''$ is the same as the midpoint $m$ of $s^*_r$, (iii) $int(t'') \neq \phi$, (iv) for every point $p \in cl(t'')$, the open segment $s_p$ has length at most $|s^*_r| + \frac{|s^*_r|}{6}$, and (v) for every point $p \in cl(t'')$, the open segment $s_p$ has length at least $\frac{|s^*_r|}{3}$ on either side of the horizontal line $h_m$, passing through $m$.

Let $N$ be a natural number such that $N > N'$ and for every natural number $j > N$, every open segment in $H(t_{i_j})$ intersects $cl(t'')$. Let $p' \in cl(t'')$ be the point such that 
$u_{i_{N+1}-1} = s_{p'}$. By the properties of $t''$, the midpoint $m_{i_{N+1}-1}$ of vertical segment $u_{i_{N+1}-1}$ lies within distance at most $\frac{|s^*_r|}{4}$ from $t''$.

Then, the horizontal segment $t_{i_{N+1}}$ constructed by $\Sw_{greedy}$ will extend at least from $u_{i_{N+1}-1}$ to $s^*_r$. This contradicts our assumption that $r$ is an infinite ray of $T(\Sw_{greedy})$.
$\blacksquare$

\begin{theorem}
$\Sw_{greedy}$ is a maximal sweepline algorithm.
\end{theorem}

\noindent {\bf Proof:} Let $K_{greedy}$ be the piecewise-vertical Jordan curve forming the boundary of the region swept by the greedy sweepline
algorithm $\Sw_{greedy}$ (see Theorem $3$ in \cite{apurva}).

We prove that the set $\V(K_{greedy})$ of vertical segments of $K_{greedy}$ is empty. Suppose, for the sake of contradiction, that $\V(K_{greedy}) \neq \phi$. Let $u$ be {\it any} vertical segment in $\V(K_{greedy})$. There are two cases: (i) $u$ is the limiting segment $s^*_r$ of an infinite ray $r$ of $T(\Sw_{greedy})$, or (ii) $u \in \V(K_N)$ for some natural number $N$.

By Lemma \ref{lem-no-inf-ntm}, case (i) cannot occur. Let us consider case (ii). Let $\delta$ be the length of vertical segment $u$.
Note that $\delta > 0$.

Since the greedy sweepline algorithm never picks the vertical segment $u$ for extension, this implies that (a) it runs forever and (b) all vertical segments $u_{N+1}, u_{N+2}, \ldots$ have length at least $\delta$.

By Observation \ref{obs-max-vertical-seg-horizontal-sweep}, each horizontal sweep $H(t_i)$, $i \in \mathbb{N}$, has a finite number of vertical segments in $bd(H(t_i))$ of length at least $\delta$. 
Thus, $T(\Sw_{greedy})$ is a locally finite tree with an infinite number of nodes. By K\H{o}nig's lemma, $T(\Sw_{greedy})$ has an infinite ray $r$, such that $r$ starts at the root node of $T(\Sw_{greedy})$.

Let $r=(H(t_{i_1}), H(t_{i_2}), \ldots$, where $i_1=1$ and $i_1 < i_2 < \cdots$. Let $M$ be a natural number such that for all $j \geq M$, $i_j > N+1$. 
For each $j > 1$, $u_{i_j-1}$ is the vertical segment used by $\Sw_{greedy}$ for extending ray $r$ by adding the horizontal sweep $H(t_{i_j})$. 

Then, $u_{i_M-1}, u_{i_{M+1}-1}, \ldots$ is an infinite sequence of open segments such that (i) each open segment in the sequence has length at least $\delta$ and
(ii) for each $j \geq M$, $u_{i_j-1} \in H(t_{i_j})$. 

Then, there exists an infinite convergent subsequence of $u_{i_M-1}, u_{i_{M+1}-1}, \ldots$. The limiting segment $s^*$ of this subsequence has length at least $\delta$. Since 
$|W(r)|=1$ (see Lemma $7$ of \cite{apurva}), we conclude that $s^*_{r}$ has length at least $\delta$. Thus, $r$ is a non-terminating infinite ray in $T(\mathcal{S}_{greedy})$. We arrive at a contradiction due to Lemma \ref{lem-no-inf-ntm} proved above.
$\blacksquare$

\end{document}